%% file: main.tex
\pdfoutput=1

\documentclass[letterpaper, 10 pt, conference]{ieeeconf}  

\IEEEoverridecommandlockouts                              

\usepackage{graphicx,algorithm,algpseudocode,amsthm,amsmath,amsfonts,verbatim,mathtools,enumerate,bm,hyperref,tikz,circuitikz,multirow,makecell,adjustbox}
\usetikzlibrary{patterns}
\usepackage{tikz,csvsimple}
\usepackage{pgfplots,pgfplotstable}
\pgfplotsset{compat=1.9}
\usepackage{filecontents}

\input{data/exp.tex}
\usepackage{siunitx}
\usetikzlibrary{positioning}
\usepackage[style=base]{caption}
\usepackage{subcaption}

\usetikzlibrary{calc}
\tikzset{>=latex}

\allowdisplaybreaks

\input defs.tex

\title{\LARGE \bf Online Poisoning Attacks Against Data-Driven Predictive Control
}

\author{
Yue Yu, Ruihan Zhao, Sandeep Chinchali, and Ufuk Topcu 
\thanks{
Y. Yu, R. Zhao, and U. Topcu are with the Oden Institute for Computational Engineering and Sciences, The University of Texas at Austin, TX, 78712, USA (emails:  yueyu@utexas.edu,\,ruihan.zhao@utexas.edu,\,utopcu@utexas.edu). S. Chinchali is with the Department of Electrical and Computer Engineering, The University of Texas at Austin, TX, 78712, USA (email: sandeepc@utexas.edu). 
}%
}

\begin{document}

\maketitle
\thispagestyle{empty}
\pagestyle{empty}

\begin{abstract}
Data-driven predictive control (DPC) is a feedback control method for systems with unknown dynamics. It repeatedly optimizes a system's future trajectories based on past input-output data. We develop a numerical method that computes poisoning attacks that inject additive perturbations to the online output data to change the trajectories optimized by DPC. This method is based on implicitly differentiating the solution map of the trajectory optimization in DPC. We demonstrate that the resulting attacks can cause an output tracking error one order of magnitude higher than random perturbations in numerical experiments.  
\end{abstract}

\input{introduction/introduction}

\input{stackelberg/stackelberg}

\input{differentiation/differentiation}
\input{numerical/numerical}
\input{conclusion/conclusion}




\bibliographystyle{IEEEtran}
\bibliography{IEEEabrv,reference}

\end{document}

%% file: defs.tex
\newcommand{\blkdiag}{\mathop{\rm blkdiag}}

\newcommand{\argmin}{\mathop{\rm argmin}}
\newcommand{\argmax}{\mathop{\rm argmax}}

\newcommand{\norm}[1]{\left\lVert#1\right\rVert}
\newcommand{\mnorm}[1]{{\left\vert\kern-0.25ex\left\vert\kern-0.25ex\left\vert #1 
    \right\vert\kern-0.25ex\right\vert\kern-0.25ex\right\vert}}

\newtheorem{definition}{Definition} 

\newtheorem{lemma}{Lemma}

\newtheorem{remark}{Remark}
\newtheorem{proposition}{Proposition}
\newtheorem{assumption}{Assumption}

\newcommand{\ie}{{\it i.e.}}

\definecolor{mygreen}{rgb}{0, 0.5, 0}
\definecolor{myyellow}{rgb}{1, 0.75, 0}

%% file: introduction/introduction.tex
\section{Introduction}
\label{sec: introduction}

Data-driven predictive control (DPC) is a feedback control method for systems with unknown dynamics \cite{coulson2019data,allibhoy2020data,dorfler2022bridging}. It combines the idea of Willems' fundamental lemma and model predictive control: the former gives a parameterization of the system's future input-output trajectories using linear functions of the past input-output data \cite{willems2005note,van2020willems,yu2021controllability}, and the latter gives a feedback controller that repeatedly optimizes the system's future input-output trajectories \cite{mayne2000constrained,mayne2014model}. DPC has been successful for various systems, including quadrotors \cite{elokda2021data}, power converters \cite{huang2021decentralized}, as well as building heating, ventilation, and air conditioning \cite{chinde2022data}.

Since DPC relies heavily on data, it is susceptible to adversarial data perturbations, or \emph{data poisoning attacks} \cite{goodfellow2014explaining,kurakin2016adversarial,sharma2019attacks,agarwal2022task}. On the other hand, it remains unclear how vulnerable DPC is against data poisoning attacks. The results in \cite{alpago2020extended,coulson2021distributionally} show that DPC is robust against zero-mean stochastic noise in data. But they do not extend to deterministic data poisoning attacks. Meanwhile, results on data poisoning attacks against state estimators \cite{pajic2016attack,miao2016coding,pajic2017design,jovanov2019relaxing} and virtual reference feedback controllers \cite{russo2021poisoning,russo2021data} do not consider DPC, or any trajectory-optimization-based controllers. To our best knowledge, data poisoning attacks against DPC have received little if any attention.  

We formulate a data-poisoning attack problem in DPC, where an \emph{attacker} computes bounded additive perturbations to the online output data---which lack the thorough validation typically available for offline data and, compared with input data, are sensitive to noninvasive modifications of the sensors' physical environment \cite{pajic2016attack}---to change the trajectories optimized by DPC. We show that computing a poisoning attack is a bilevel optimization: the lower level optimizes the trajectory in DPC, and the upper level optimizes the attack. 

Furthermore, we develop an efficient numerical method to compute poisoning attacks against DPC by approximating the bilevel optimization using a single-level convex optimization. We construct this approximation in two steps. First, we transform the bilevel optimization into a single-level nonconvex optimization using the solution map of the lower-level trajectory optimization. Second, we approximate the single-level nonconvex optimization using a convex one by implicitly differentiating said solution map.

Finally, we demonstrate the effectiveness and efficiency of the proposed method in attacking DPC for a linear oscillating masses system and a nonlinear quadrotor system in PyBullet, a high-fidelity robotics simulator \cite{panerati2021learning}. Our numerical experiments show that the performance of DPC is more sensitive to data-poisoning attacks than random noise: the former can cause an output tracking error one order of magnitude higher than that of the latter. Furthermore, the proposed method is more efficient in implicit differentiation than CVXPYlayers, a state-of-the-art differentiation toolbox \cite{agrawal2019differentiable}. In our experiments, the least-squares problem solved in the proposed method---which is the main computational task in implicit differentiation---is about half the size of the one solved in CVXPYlayers. 

Our work complements the qualitative results in robust DPC. Although regularization can qualitatively stabilize DPC against data perturbations \cite{berberich2020data,berberich2022linear}, how to quantitatively compute the regularization parameters is, to our best knowledge, still an open question. Our results enable quantitative evaluation and potentially automated search of these regularization parameters based on their stabilizing performance under poisoning attacks.

\paragraph*{Notation} We let \(\mathbb{R}\), \(\mathbb{R}_+\), \(\mathbb{R}_{++}\), and \(\mathbb{N}\) denote the set of real, nonnegative real, positively real, and positive integer numbers, respectively. Given \(m, n\in\mathbb{N}\), we let \(\mathbb{R}^n\) and \(\mathbb{R}^{m\times n}\) denote the set of \(n\)-dimensional real vectors and \(m\times n\) real matrices, We let \(0_n\)\(\) and \(I_n\) denote the \(n\)-dimensional zero vector and the \(n\times n\) identity matrix, respectively. Given a square real matrix \(A\in\mathbb{R}^{n\times n}\), we let \(A^\top\), \(A^{-1}\), and \(A^\dagger\) denote the transpose, the inverse, and the Moore–Penrose inverse of matrix \(A\), respectively. Given a symmetric and positive semidefinite matrix \(M\in\mathbb{R}^{n\times n}\) and \(x\in\mathbb{R}^n\), we let \(\norm{x}\coloneqq \sqrt{x^\top x}\) and \(\norm{x}_M\coloneqq \sqrt{x^\top M x}\). We let \(\partial G(x)\in\mathbb{R}^{m\times n}\) denote the Jacobian matrix of function \(G\) evaluated at \(x\in\mathbb{R}^n\). We say a closed set \(\mathbb{D}\subset\mathbb{R}^n\) is a closed convex set if \(\alpha x+(1-\alpha)y\in\mathbb{D}\) for all \(x, y\in\mathbb{D}\) and \(\alpha\in[0, 1]\).
The \emph{projection} of \(x\in\mathbb{R}^n\) onto a closed convex set \(\mathbb{D}\subset\mathbb{R}^n\) is a function \(\Pi_{\mathbb{D}}:\mathbb{R}^n\to\mathbb{D}\) where \(
    \Pi_{\mathbb{D}}(x)\coloneqq \underset{x'\in\mathbb{D}}{\mbox{argmin}}\, \norm{x'-x}\).

%% file: stackelberg/stackelberg.tex
\section{Online poisoning attack problem in data-driven predictive control}
\label{sec: stackelberg}

We introduce the data poisoning attack problem in data-driven predictive control (DPC). We will first revisit the basics of DPC, then introduce a bilevel optimization that models of data poisoning attacks against DPC.

\subsection{Data-driven predictive control}
We will briefly review the basics of
data-driven predictive control, a control law for unknown dynamical systems based on data and optimization.

\subsubsection{Input and output trajectories}
Consider a discrete time dynamical system with \(n_u\) inputs and \(n_y\) outputs. The system's input and output at time \(j\in\mathbb{N}\) are denoted by \(u_j\in\mathbb{R}^{n_u}\) and \(y_j\in\mathbb{R}^{n_y}\), respectively. 

At each sampling time \(k\), DPC requires the knowledge of an online input-output trajectories generated by the system, denoted by \(\{u_{k-\sigma}, \ldots, u_{k-1}\}\) and \(\{y_{k-\sigma},\ldots y_{k-1}\}\), where \(\sigma\in\mathbb{N}\) is the upper bound of the lag of the system \cite{coulson2019data}. Intuitively, \(\sigma\) is the number of input-output pairs needed to pinpoint the state of the system. In addition, DPC also requires the knowledge of the input Hankel matrix \(U\in\mathbb{R}^{(\ell+\sigma)n_u\times n_g}\) and output Hankel matrix \(Y\in\mathbb{R}^{(\ell+\sigma)n_y\times n_g}\), where \(\ell\in\mathbb{N}\) is the planning horizon in DPC, and \(n_g\in\mathbb{N}\) is determined by the amount of offline data. See \cite{van2020willems,yu2021controllability} for a detailed discussion on constructing Hankel matrices using offline input-output trajectories.

\subsubsection{Data-driven trajectory optimization}
We now introduce the trajectory optimization problem used in data-driven predictive control. To this end, we first introduce the following notation:
\begin{equation}\label{eqn: Hankel}
    \begin{aligned}
        &u_{\text{ini}}\coloneqq \begin{bsmallmatrix}
        u_{k-\sigma} \\  \vdots \\ u_{k-1}
        \end{bsmallmatrix},  y_{\text{ini}}\coloneqq \begin{bsmallmatrix}
        y_{k-\sigma} \\ \vdots \\ y_{k-1}
        \end{bsmallmatrix},  \begin{bmatrix}
        U_{\text{p}}\\
        U_{\text{f}}
        \end{bmatrix} =U, \begin{bmatrix}
        Y_{\text{p}}\\
        Y_{\text{f}}
        \end{bmatrix} =Y,
    \end{aligned}
\end{equation}
where \(U_{\text{p}}\in\mathbb{R}^{\sigma n_u\times n_g}\) and \(U_{\text{f}}\in\mathbb{R}^{\ell n_u\times n_g}\)  are partitions of \(U\),  \(Y_{\text{p}}\in\mathbb{R}^{\sigma n_y\times n_g}\) and \(Y_{\text{f}}\in\mathbb{R}^{\ell n_y\times n_g}\) are partitions of \(Y\). 

At each discrete time \(k\in\mathbb{N}\), the data-driven predictive controller computes a length-\(\ell\) future input and output trajectory of the system--- denoted by \(u\coloneqq \begin{bsmallmatrix} 
    u_k^\top & u_{k+1}^\top & \cdots & u_{k+\ell-1}^\top \end{bsmallmatrix}^\top\) and \(y\coloneqq \begin{bsmallmatrix} 
    y_k^\top & y_{k+1}^\top & \cdots & y_{k+\ell-1}^\top \end{bsmallmatrix}^\top\),
respectively---by solving the following optimization problem:
\begin{equation}\label{opt: DeePC}
    \begin{array}{ll}
        \underset{u, y,
        g}{\mbox{minimize}} &  \frac{1}{2}\norm{y-\hat{y}}^2_Q+\frac{1}{2}\norm{u-\hat{u}}^2_R+\lambda_g \norm{Mg}^2\\
        &+\lambda_s\norm{Y_{\text{p}}g-y_{\text{ini}}}^2\\
         \mbox{subject to} & u_{\text{ini}}=U_{\text{p}}g,\, u = U_{\text{f}} g,\, y = Y_{\text{f}} g,\, u\in\mathbb{U}, \, y\in\mathbb{Y}.
        \end{array}
\end{equation}
where
\(
    M\coloneqq I_{n_g}-\begin{bsmallmatrix}
U_p\\
Y_p\\
U_f
\end{bsmallmatrix}^\dagger \begin{bsmallmatrix}
U_p\\
Y_p\\
U_f
\end{bsmallmatrix}
\) is a weighting matrix inspired by system identification and has proven to be more effective than identity weighting \cite{dorfler2022bridging} ; matrix \(Q\in\mathbb{R}^{\ell n_u\times \ell n_u}\) and matrix \(R\in\mathbb{R}^{\ell n_y\times \ell n_y}\) are both symmetric and positive semidefinite; \(\lambda_g, \lambda_s\in\mathbb{R}_+\) are regularization weights; set \(\mathbb{U}\subset\mathbb{R}^{\ell n_u}\) and set \(\mathbb{Y}\subset\mathbb{R}^{\ell n_y}\) are the feasible set of input and output trajectories, respectively; \(\hat{u}\in\mathbb{R}^{\ell n_u}\) and \(\hat{y}\in\mathbb{R}^{\ell n_y}\) are the reference input and output trajectory in DPC, respectively. The constraints and objective function in optimization~\eqref{opt: DeePC} ensures that \(\begin{bsmallmatrix}
u_{\text{ini}}\\
u
\end{bsmallmatrix}=\begin{bsmallmatrix}
U_p\\
U_f
\end{bsmallmatrix}g\) and \(\begin{bsmallmatrix}
y_{\text{ini}}\\
y
\end{bsmallmatrix}\approx\begin{bsmallmatrix}
Y_p\\
Y_f
\end{bsmallmatrix}g\), which says all trajectories are approximately linear functions of past data.

\subsection{Poisoning attacks against data-driven predictive control}

We consider a scenario where the online output measurements in optimization~\eqref{opt: DeePC} are corrupted by bounded additive perturbations designed by a malicious attacker. To this end, we start with the following variation of optimization~\eqref{opt: DeePC}:
\begin{equation}\label{opt: DeePC w/ poison}
    \begin{array}{ll}
        \underset{u, y,
        g}{\mbox{minimize}} &  \frac{1}{2}\norm{y-\hat{y}}^2_Q+\frac{1}{2}\norm{u-\hat{u}}^2_R+\lambda_g \norm{Mg}^2\\
        &+\lambda_s\norm{Y_{\text{p}}g-(y_{\text{ini}}+p)}^2\\
         \mbox{subject to} & u_{\text{ini}}=U_{\text{p}}g,\, u = U_{\text{f}} g,\, y = Y_{\text{f}} g,\, u\in\mathbb{U}, \, y\in\mathbb{Y},
        \end{array}
\end{equation}
where \(p\in\mathbb{R}^{\sigma n_y}\) is an attacking perturbation. Notice that optimization~\eqref{opt: DeePC} is a special case of optimization \eqref{opt: DeePC w/ poison} if \(p=0_{\sigma n_y}\). 
We focus on attacks against the online output data in vector \(y_{\text{ini}}\). Unlike the offline data in matrix \(Y_p\) and \(Y_f\), these data are generated in real-time and are more likely to lack thorough validation. Unlike input data, they are sensitive to noninvasive modification of sensors' physical  environment, even when properly encrypted \cite{pajic2016attack}. 

We now introduce the poisoning attack problem against data-driven predictive control, where an attacker seeks the optimal bounded perturbation \(p\) such that the optimal input trajectory in optimization~\eqref{opt: DeePC w/ poison} minimizes a performance function chosen by the attacker. We summarize the definition of this poisoning attack problem as follows. 

\begin{definition}[Poisoning attack problem]\label{def: poison}
Given optimization~\eqref{opt: DeePC w/ poison}, a continuously differentiable cost function \(\psi:\mathbb{R}^{\ell n_u}\to\mathbb{R}\) that evaluates the performance of the attacked trajectory, and a closed convex set \(\mathbb{P}\subset\mathbb{R}^{\sigma n_y}\) for admissible attacking perturbations, the poisoning attack problem seeks the optimal perturbation vector \(p\in\mathbb{P}\) in the following bilevel optimization problem:
\begin{equation}\label{opt: poison}
    \begin{array}{ll}
        \underset{u, y, g, p}{\mbox{minimize}} & \psi(u) \\
        \mbox{subject to} &p\in\mathbb{P},\, \begin{bmatrix} u^\top & y^\top & g^\top\end{bmatrix}^\top \text{ is optimal for \eqref{opt: DeePC w/ poison}.}
    \end{array}
\end{equation}
\end{definition}

As an example of problem~\eqref{opt: poison}, one can let \begin{equation}\label{eqn: psi P}
    \textstyle \psi(u)=\frac{1}{2}\norm{u-\tilde{u}}^2, \enskip \mathbb{P}=\{p\in\mathbb{R}^{\sigma n_y}|\norm{p}\leq \rho\norm{y_{\text{ini}}}\},
\end{equation}
where \(\tilde{u}\) is the attacker's desired input trajectory, and \(\rho\in\mathbb{R}_+\) is the ratio between the norm of the attacking perturbation and the output measurements. In this case, the attacker aims to push the input trajectory computed by DPC towards \(\tilde{u}\) by adding a perturbation \(p\), where the perturbation-to-data ratio is upper bounded by \(\rho\). By choosing different values of \(\rho\), one can evaluate the effects of different attacks against DPC by solving optimization~\eqref{opt: poison}.

We assume that the attacker's objective function only depends on the input trajectory \(u\) and that it has full knowledge of the parameters in optimization~\eqref{opt: DeePC} (such as \(\lambda_s\) and \(\lambda_g\)). The former assumption is because DPC uses only the input trajectory \(u\) to construct the input to the system \cite[Alg. 2]{coulson2019data}. The output trajectory \(y\), on the other hand, is merely a computational byproduct during this construction. The latter assumption ensures that problem~\eqref{opt: DeePC w/ poison} gives the worst-case estimate of the attacker's perturbations. 

If the underlying system is linear time-invariant and the output data are the state data, then, due to Willems' lemma \cite{willems2005note,van2020willems,yu2021controllability},  problem~\eqref{opt: DeePC w/ poison} is equivalent to the attacking of model predictive control. However, problem~\eqref{opt: DeePC w/ poison} also applies to nonlinear systems and systems without state data.


%% file: differentiation/differentiation.tex
\section{Poisoning attacks via implicit differentiation}
\label{sec: implicit}

We introduce an efficient numerical method to approximately solve the bilevel optimization problem in \eqref{opt: poison}. Our method is based on the implicit function theorem \cite[Thm. 1B.1]{dontchev2014implicit} and a novel form of optimality conditions based on the Minty parameterization theorem \cite[Prop. 23.22]{bauschke2017convex}.

To simplify our notation in this section, we will first rewrite trajectory optimization~\eqref{opt: DeePC w/ poison} in a compact form. To this end, we introduce the following notation:
\begin{equation}\label{eqn: param}
\begin{aligned}
   m & = \ell (n_u+ n_y)+\sigma n_u,\enskip n  = \ell (n_u+ n_y)+n_g,\\
   z & =\begin{bmatrix}
    u^\top & y^\top & g^\top
    \end{bmatrix}^\top, \enskip \mathbb{D}  = \mathbb{U}\times \mathbb{Y}\times \mathbb{R}^{n_g},\\
P&=\blkdiag(R, Q, 2\lambda_g M^\top M+ 2\lambda_s Y_{\text{p}}^\top Y_{\text{p}} ),\\
q(p)&=-\begin{bmatrix}
R\hat{u}^\top & Q\hat{y}^\top & 2\lambda_s (p+y_{\text{ini}})^\top Y_{\text{p}}
\end{bmatrix}^\top,\\
    H & = \begin{bmatrix}
    0_{\sigma n_u\times \ell n_u} & 0_{\sigma n_u\times \ell n_y} & U_{\text{p}} \\
    -I_{\ell n_u} & 0_{\ell n_u\times \ell n_y} & U_{\text{f}} \\
    0_{\ell n_y\times \ell n_u} & -I_{\ell n_y} & Y_{\text{f}} 
    \end{bmatrix}, \enskip    b  = \begin{bmatrix}
    u_{\text{ini}} \\ 0_{\ell( n_u+ n_y)}
    \end{bmatrix},
\end{aligned}
\end{equation}
 where \(P\) is the block diagonal matrix obtained by aligning \(R\) \(Q\), \(2\lambda_g M^\top M+2\lambda_s Y_{\text{p}}^\top Y_{\text{p}}\) along its diagonal.

With the above notation, we can rewrite optimization~\eqref{opt: DeePC w/ poison} in the following form:
\begin{equation}\label{opt: conic}
    \begin{array}{ll}
    \underset{z}{\mbox{minimize}}  & \frac{1}{2}z^\top Pz+q(p)^\top z\\
    \mbox{subject to} & Hz=b,\, z\in\mathbb{D}.
    \end{array}
\end{equation}

Next, we will work with the compact notation in \eqref{opt: conic}, rather than \eqref{opt: DeePC w/ poison}; we remind the readers again that the two are exactly equivalent due to \eqref{eqn: param}.

Throughout, we will make the following assumption on optimization~\eqref{opt: conic}. 
\begin{assumption}\label{asp: KKT} Set \(\mathbb{D}\subseteq \mathbb{R}^n\) is closed and convex. There exists \(z^\star\in\mathbb{R}^n\) and \(w^\star\in\mathbb{R}^m\) such that
\begin{equation}\label{eqn: minmax}
    z^\star\in\, \underset{z\in\mathbb{D}}{\argmin}\, L(z, w^\star),\enskip w^\star\in\, \underset{w}{\argmax}\, L(z^\star, w),
\end{equation}
where \(L(z, w)\coloneqq \frac{1}{2}z^\top Pz+q(p)^\top z+w^\top (Hz-b)\).
\end{assumption}

Under mild constraint qualification conditions on optimization~\eqref{opt: conic} \cite[Cor. 28.3.1]{rockafellar2015convex}, \eqref{eqn: minmax} holds if and only if there exists an optimal solution for optimization~\eqref{opt: conic}. 
 
We are interested in how the perturbation vector \(p\) affects the optimal solutions of optimization~\eqref{opt: conic}.  To this end, we introduce the following definition, which characterizes the mathematical relations between the two. 

\begin{definition}[Solution map and its localization]\label{def: sol map} The solution map of optimization~\eqref{opt: conic} is denoted by
\(S:p\mapsto S(p)\) where \(S(p)\coloneqq \{\begin{bmatrix}
    (z^\star)^\top & (w^\star)^\top
    \end{bmatrix}^\top | \text{Conditions in \eqref{eqn: minmax} hold.} \}\). We say solution map \(S\) has a single-valued localization around \(p\) if there exists a function \(\tilde{S}\) such that \(\tilde{S}(\tilde{p})\in S(\tilde{p})\) for all \(\tilde{p}\) in a neighborhood of \(p\).
\end{definition} 
 
In the following, we will discuss the differentiability properties of solution map \(S\) and its localization.

\subsection{Optimality conditions as nonlinear equations}

The implicit function theorem provides a characterization of the Jacobian of the solution maps of nonlinear equations \cite[Thm. 1B.1]{dontchev2014implicit}. On first look, the implicit function theorem seems not applicable to the solution map \(S\) in Definition~\ref{def: sol map}, since the latter is defined by an optimization problem rather than nonlinear equations. However, the following proposition shows that the optimality conditions in \eqref{eqn: minmax} are actually equivalent to a set of nonlinear equations, laying the groundwork for applying the implicit function theorem to the solution map \(S\).

\begin{lemma}\label{lem: Minty}
Let 
\begin{equation*}
    F\left(\begin{bmatrix}
z^\top & w^\top
\end{bmatrix}^\top, p\right) \coloneqq \begin{bmatrix}
    z-\Pi_{\mathbb{D}}(z -Pz-q(p)-H^\top w) \\
    Hz-b
\end{bmatrix},
\end{equation*}
for all \(z\in\mathbb{R}^n\) and \(w\in\mathbb{R}^m\).
Then the conditions in \eqref{eqn: minmax} are equivalent to the following conditions: \begin{equation}\label{eqn: saddle proj}
    F\left(\begin{bmatrix}
(z^\star)^\top & (w^\star)^\top
\end{bmatrix}^\top, p\right) = 0_{m+n}.
\end{equation}
\end{lemma}

\begin{proof}
First, the conditions in \eqref{eqn: minmax} are equivalent to the following \cite[Thm. 27.4]{rockafellar2015convex}:
\begin{equation}\label{eqn: saddle}
    Pz^\star+q(p)+H^\top w^\star +N_{\mathbb{D}}(z^\star)  \ni 0_n,\enskip Hz^\star-b =0_m.
\end{equation}
where \(N_{\mathbb{D}}(z^\star)\)is the normal cone of set \(\mathbb{D}\) at \(z^\star\). Next, since the set \(\mathbb{D}\) is nonempty, closed, and convex, \(N_{\mathbb{D}}\) is a maximal monotone operator \cite[Ex. 20.26]{bauschke2017convex}. The rest of the proof follows directly from the Minty parameterization theorem for maximal monotone operators \cite[Prop. 23.22]{bauschke2017convex}.
\end{proof}

Lemma~\ref{lem: Minty} shows that the optimizers in \eqref{eqn: minmax} are the fixed-points of the \emph{proportional integral projected gradient method}, a first-order primal-dual conic optimization method that combines gradient descent with proportional-integral feedback \cite{yu2020proportional,yu2022proportional,yu2023extrapolated}. 

\begin{remark}
Compared with the results in \cite{agrawal2019differentiable}, the optimality conditions in Lemma~\ref{lem: Minty} contain fewer variables by eliminating dual variables for any inequality constraints.
\end{remark}


\subsection{Implicit differentiation through optimality conditions}

Equipped with Lemma~\ref{lem: Minty}, we are ready to present the differentiability properties of the solution map in Definition~\ref{def: sol map} as follows.

\begin{proposition}[Implicit function theorem]\label{prop: implicit}
Consider the solution mapping \(S\) in Definition~\ref{def: sol map}. Suppose Assumption~\ref{asp: KKT} holds and \(z^+ \coloneqq  z^\star-Pz^\star-q-H^\top w^\star\). In addition, suppose function \(\Pi_{\mathbb{D}}\) is continuously differentiable within a neighborhood of \(z^+\). Let
\begin{equation}\label{eqn: Jacobian}
\begin{aligned}
    J\coloneqq & \begin{bmatrix}
    I_n-\partial\Pi_{\mathbb{D}}(z^+)(I_n- P) &  \partial\Pi_{\mathbb{D}}(z^+)H^\top \\
    H & 0_{m\times m}
    \end{bmatrix},\\
    K\coloneqq & \begin{bmatrix}
    \partial\Pi_{\mathbb{D}}(z^+)\begin{bmatrix}
    0_{\sigma n_y\times (\ell n_u+\ell n_y)} & -2\lambda_s Y_{\text{p}}
    \end{bmatrix}^\top \\
    0_{m\times \sigma n_y}
    \end{bmatrix}.
\end{aligned}
\end{equation}
If matrix \(J\) is nonsingular, the solution map \(S\) has a single-valued localization \(\tilde{S}\). Within a neighborhood of \(p\), function \(\tilde{S}\) is continuously differentiable with its Jacobian satisfying:
\begin{equation}
    \partial  \tilde{S}(p)  =-J^{-1}K.
\end{equation}
\end{proposition}

\begin{proof}
Our proof is based on the implicit function theorem \cite[Thm. 1B. 1]{dontchev2014implicit}. To use this theorem, let \(\xi\in\mathbb{R}^{m+n}\) be arbitrary and consider the nonlinear equations \(F(\xi, p)=0_{m+n}\) in \eqref{eqn: saddle proj} and the Jacobian \(\partial_{\xi}F(\xi, p)\) and \(\partial_{p}F(\xi, p)\).
Using the chain rule we can show that \(\partial_{\xi}F(\xi, p)=J\) and \(\partial_{p}F(\xi, p)=K\), where \(J\) and \(K\) are given in \eqref{eqn: Jacobian}. The rest of the proof is a direct application of the implicit function theorem \cite[Thm. 1B. 1]{dontchev2014implicit}. 
\end{proof}

\begin{remark}
Proposition~\ref{prop: implicit} assumes the local differentiability of function \(\Pi_{\mathbb{D}}\). This assumption holds almost everywhere if set \(\mathbb{U}\) and \(\mathbb{Y}\) are Cartesian products of many common closed convex cones; see  \cite[Sec. 3]{busseti2019solution} for an overview. If these sets are intervals, then \(\Pi_{\mathbb{D}}\) is piecewise-linear \cite[Thm. 3.3.14]{bauschke1996projection} and also differentiable almost everywhere.
\end{remark}

\subsection{Poisoning attacks via implicit differentiation}

We now introduce the \emph{approximate poisoning attack problem}, where we approximate the bilevel optimization in \eqref{opt: poison} via linearization and implicit differentiation. 

Due to Definition~\ref{def: sol map}, we know that \(z\in\mathbb{R}^n\) is optimal for \eqref{opt: poison} if and only if there exists \(w\in\mathbb{R}^m\) such that \(\begin{bmatrix} z^\top & w^\top\end{bmatrix}^\top\in S(p)\). Hence \(\begin{bmatrix} u^\top & y^\top & g^\top\end{bmatrix}^\top\) is optimal for optimization~\eqref{opt: poison} if and only if \(u\in TS(p)\), where  \(T=\begin{bmatrix}
I_{\ell n_u} & 0_{\ell n_u\times (m+n-\ell n_u)}
\end{bmatrix}\). Therefore, we 
we can rewrite optimization~\eqref{opt: poison} equivalently as the following one:

\begin{equation}\label{opt: poison w/ map}
    \begin{array}{ll}
        \underset{p\in\mathbb{P}}{\mbox{minimize}} & \psi(TS(p)).
    \end{array}
\end{equation}

Second, suppose that the assumptions in Proposition~\ref{prop: implicit} hold when \(p=0_{\sigma n_y}\). Then, there exists unique \(\xi^\star\coloneqq\begin{bmatrix} (z^\star)^\top & (w^\star)^\top\end{bmatrix}^\top\in\mathbb{R}^{m+n}\) such that \(
    \xi^\star= \tilde{S}(0_{m+n})\),
where \(\tilde{S}\) is the single-valued localization of \(S\) around \(0_{\sigma n_y}\).
Due to the chain rule and Proposition~\ref{prop: implicit}, the following approximation based on Talor series holds for all \(p\approx 0_{\sigma n_y}\):
\begin{equation}\label{eqn: linearization}
\begin{aligned}
     \psi\big(TS(p) \big)=&\psi\big(T\tilde{S}(p) \big)\approx \psi \big(T\xi^\star\big)-\partial\psi\big(T\xi^\star\big) TJ^{-1}Kp,
\end{aligned}
\end{equation}
where matrix \(J\) and \(K\) are given by \eqref{eqn: Jacobian} with
\begin{equation}
    q=-\begin{bmatrix}
R\hat{u}^\top & Q\hat{y}^\top & 2\lambda_s y_{\text{ini}}^\top Y_{\text{p}}
\end{bmatrix}^\top.
\end{equation}
In other words, we let \(p=0_{\sigma n_y}\) in \eqref{eqn: param}. 

By substituting the linear approximation in \eqref{eqn: linearization} into optimization~\eqref{opt: poison w/ map}, we obtain the following \emph{approximate poisoning attack problem}:
\begin{equation}\label{opt: poison approx}
    \begin{array}{ll}
        \underset{p\in\mathbb{P}}{\mbox{minimize}} & - \partial\psi(T\xi^\star) T J^{-1} K p
    \end{array}
\end{equation}

\begin{algorithm}[!ht]
\caption{Poisoning attack via implicit differentiation}
\begin{algorithmic}[1]
\Require The parameters in optimization~\eqref{opt: conic}. 
\State Solve optimization~\eqref{opt: conic} with \(p=0_{\sigma n_y}\) for \(z^\star\) and \(w^\star\) such that \(\xi^\star=\begin{bmatrix} (z^\star)^\top & (w^\star)^\top\end{bmatrix}^\top\) and \(F\big(\xi^\star, 0_{\sigma n_y}\big)=0_{m+n}\).
\State Compute \(J\) and matrix \(K\) using \eqref{eqn: Jacobian} with \(p=0_{\sigma n_y}\).
\State \(p^\star\in\underset{z\in\mathbb{P}}{\argmin}\, - \partial \psi(T\xi^\star) T J^{-1} K p\).\label{opt: cond grad}
\Ensure Perturbation \(p^\star\).
\end{algorithmic}
\label{alg: attack}
\end{algorithm}

Based on the observations above, we present Algorithm~\ref{alg: attack} for computing an approximate solution to the poisoning attack problem in Definition~\ref{def: poison}. Notice that, since matrix \(J\) in \eqref{opt: poison approx} can be singular---or ill-conditioned---we approximate the matrix inverse in \eqref{opt: poison approx} with the corresponding Moore–Penrose pseudoinverse in line~\ref{opt: cond grad}. Evaluating this approximation requires solving a least-squares problem, which is common in implicit differentiation \cite{agrawal2019differentiable}. 

Using Algorithm~\ref{alg: attack}, one can approximately solve bilevel optimization \eqref{opt: poison}, which is NP-hard to solve exactly \cite{bard2013practical}. Implementing Algorithm~\ref{alg: attack} only requires solving the convex trajectory optimization problem in \eqref{opt: conic}, a least-squares problem, and the minimization of a linear function over set \(\mathbb{P}\). Furthermore, Proposition~\ref{prop: implicit} ensures that, under certain local differentiability and nonsingularity assumptions, Algorithm~\ref{alg: attack} gives the optimal solution of a linear approximation of the bilevel optimization in \eqref{opt: poison}.

%% file: numerical/numerical.tex
\section{Numerical experiments}
\label{sec: numerical}

We demonstrate the effectiveness of the perturbations computed by Algorithm~\ref{alg: attack} in the attacking of two different dynamical systems: a linear oscillating masses system, a popular benchmark linear system in optimal control \cite{yu2022proportional}, and the nonlinear quadrotor dynamics in robotics simulator PyBullet \cite{panerati2021learning}.

\subsection{DPC trajectory optimization setup}
Furthermore, based on the observations made in DPC literature \cite{elokda2021data,dorfler2022bridging}, we choose \(n_g=500\), \(\sigma=6\), \(\ell=25\), \(\lambda_s=10^6\), and \(\lambda_g=100\) in optimization~\eqref{opt: DeePC w/ poison}. Throughout we solve optimization~\eqref{opt: DeePC w/ poison} using ECOS \cite{domahidi2013ecos} with these parameters. We consider the following two systems.

\paragraph{Oscillating masses system}

We consider the following oscillating-masses system:\[
    x_{k+1}=\exp\left(\Delta A\right)x_k+\textstyle \int_{0}^\Delta \exp\left(sA\right)\mathrm{d}s Bu_k,\enskip y_k =  x_k,\]
where \(A=\begin{bsmallmatrix}
    0 & 0 & 1 & 0\\
    0 & 0 & 0 & 1\\
    -2 & 1 & 0 & 0\\
    1 & -2 & 0 & 0
    \end{bsmallmatrix}\), \(B=\begin{bsmallmatrix}
    0 & 0\\
    0 & 0\\
    1 & 0\\
    0 & 1
    \end{bsmallmatrix}\), \(\exp\) denotes the matrix exponential, and \(\Delta=0.1\) is the sampling time period. This system has \(n_u=2\) and \(n_y=4\), and describes the dynamics of two unit masses connected by springs with unit spring constants; see Fig.~\ref{fig: two systems} for an illustration. When solving optimization~\eqref{opt: DeePC w/ poison} with this system, we let \(Q=10I_{4}\), \(R=I_2\), and \(\hat{y}\) denote the stationary point with positive unit displacement (see Fig.~\ref{fig: masses} for an illustration). We choose \(\mathbb{U}\) and \(\mathbb{Y}\) in \eqref{opt: DeePC w/ poison} such that the input and output are elementwise bounded within the interval \([-1, 1]\) and \([-5, 5]\), respectively.

\begin{figure}[!hbt]
\centering
  \begin{subfigure}[b]{0.49\columnwidth}
  \centering
  \input{numerical/figs/mass}
  \caption{Oscillating masses}
  \end{subfigure}
  \hfill
  \begin{subfigure}[b]{0.49\columnwidth}
  \centering
  \includegraphics[trim=0cm 0cm 0cm 1.5cm,width=0.8\textwidth]{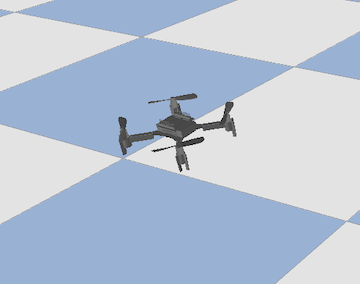}
  \caption{PyBullet quadrotor}
  \end{subfigure} 
  \caption{The two systems controlled by DPC.}
  \label{fig: two systems}
\end{figure}
\paragraph{PyBullet quadrotor system}
We also consider the nonlinear quadrotor system in PyBullet \cite{panerati2021learning}. This system does not have an explicit description via differential equations. Instead, it is a black-box system composed of pre-tuned PID controllers and 6-degree-of-freedom (6DoF) quadrotor dynamics; see Fig.~\ref{fig: quad system} for an illustration. The input of the system is a \(3\)-dimensional velocity command vector, and the output of the system is a \(6\)-dimensional vector containing the position and orientation angles of the quadrotor. Both input and output are measured at every \(0.04\) second. When solving optimization~\eqref{opt: DeePC w/ poison} with this system, we let \(Q=\blkdiag(10I_3, 0_{3\times 3})\)---\ie, positive weights on position outputs, zero weights on attitude outputs---and \(R=I_3\). We let \(\hat{y}\) denote a circular trajectory in the xy-plane with radius \(0.3\); see Fig.~\ref{fig: quad} for an illustration. In addition, we choose \(\mathbb{U}\) and \(\mathbb{Y}\) in \eqref{opt: DeePC w/ poison} such that the input is elementwise bounded within the interval \([-1, 1]\), the position output is bounded within the interval \([-2, 2]\), \([-2, 2]\), and \([0, 2]\) along the x-axis, y-axis, and z-axis, respectively; the attitude angle output is elementwise bounded within the interval \([-\frac{\pi}{8}, \frac{\pi}{8}]\). 

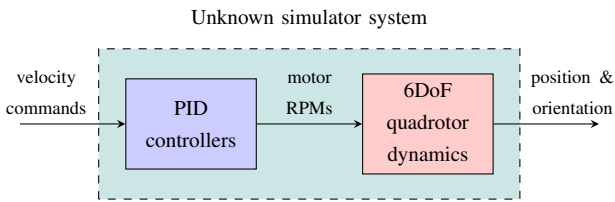
\begin{figure}[!ht]
    \centering
    \input{numerical/quadrotor}
    \caption{The structure of the nonlinear system for quadrotor dynamics in the PyBullet simulator.}
    \label{fig: quad system}
\end{figure}
\subsection{Poisoning attacks setup}

We construct the poisoning attack problem in \eqref{opt: poison} where  function \(\psi(u)\) and set \(\mathbb{P}\) are chosen according to \eqref{eqn: psi P}. For the oscillating masses system, we choose \(\tilde{u}\) to be sinusoidal signals with unit amplitude and unit angular frequency; for the PyBullet quadrotor system, we choose \(\tilde{u}\) to be a velocity command in the positive \(y\) direction.

\subsection{Numerical results}

We demonstrate the output tracking error---the difference between the system's output  trajectory and the reference trajectory---of the oscillating masses system and PyBullet quadrotor system under different poisoning attacks in Fig.~\ref{fig: masses} and Fig.~\ref{fig: quad}, respectively. 
These trajectories are simulated using the input optimized by solving \eqref{opt: DeePC w/ poison} every \(10\) steps with attacks added to output data \(y_{\text{ini}}\). As a benchmark in numerical experiments, we also consider random perturbation as \(p_{\text{rand}}=(\rho\norm{y_{\text{ini}}}/\norm{v}) v\)
where each element in \(v\in\mathbb{R}^{\sigma n_y}\) is sampled independently from the standard Gaussian distribution. From these results, we can observe that: 1) compared with the linear oscillating masses, DPC is more sensitive to perturbation when applied to the nonlinear PyBullet quadrotor system, and 2) when applied to the PyBullet quadrotor system, the perturbations computed by Algorithm~\ref{alg: attack}  increase the tracking error by more than ten times higher than those caused by random perturbations.

We also compare the efficiency of the implicit differentiation step used in Algorithm~\ref{alg: attack} against off-the-shelf software CVXPY layer \cite{agrawal2019differentiable} in terms of the dimension of the least squares problem---in Algorithm~\ref{alg: attack}, this problem is solved when evaluating the pseudoinverse in line~\ref{opt: cond grad}---they solve, which are summarized in Tab.~\ref{tab: ls size}. These results show that, since we do not rely on homogeneous self-dual embedding, the size of the least-squares problem we solve in Algorithm~\ref{alg: attack}---which equals \(2\ell(n_u+n_y)+\sigma n_u+n_g\), \ie, the number of equations in \eqref{eqn: saddle proj})---is less than half the size of the one solved in CVXPYlayers, which is around \(5\ell(n_u+n_y)+\sigma n_u+2n_g\), \ie, the size of the homogeneous self-dual embedding (HSDE) for optimization~\eqref{opt: DeePC w/ poison}\footnote{Since CVXPYlayers introduce some internal auxiliary variables, the size of the HSDE is slightly higher than \(5\ell(n_u+n_y)+\sigma n_u+2n_g\).}. 

\input{numerical/figs/plot1.tex}
\input{numerical/figs/plot2.tex}

\begin{table}[!ht]
\caption{The size of the least-squares problem solved when differentiating optimization~\eqref{opt: DeePC w/ poison}.}
\centering
   \begin{tabular}{ c|c|ccc } 
\hline
\multicolumn{2}{c|}{Trajectory length \(\ell\)} & 25 & 50 & 100\\
\hline
\multirow{2}{*}{Oscillating  masses} &
CVXPY layer & 1799 & 2549 & 4049\\
\cline{2-5}
& Algorithm~\ref{alg: attack} & \textbf{812} & \textbf{1112} & \textbf{1712}\\
\hline
\multirow{2}{*}{PyBullet quadrotor} &
CVXPY layer & 2192 & 3317 & 5567\\
\cline{2-5}
& Algorithm~\ref{alg: attack}& \textbf{968} & \textbf{1418} & \textbf{2318}\\
\hline
\end{tabular}

\label{tab: ls size}
\end{table}

%% file: numerical/figs/mass.tex
\begin{adjustbox}{scale=0.5}

\begin{circuitikz}

\pattern[pattern=north east lines] (-0.8,-0.2) rectangle (0.8, 0);
\draw[thick] (-0.8, 0) -- (0.8, 0);
\draw (0, 0) to[spring] (0, 1);
\draw[thick] (-0.5, 1) rectangle (0.5, 2);
\draw (0, 2) to[spring] (0, 3);
\draw[thick] (-0.5, 3) rectangle (0.5, 4);
\draw (0, 4) to[spring] (0, 5);
\pattern[pattern=north east lines] (-0.8, 5) rectangle (0.8, 5.2);
\draw[thick] (-0.8, 5) -- (0.8, 5);

\end{circuitikz}

\end{adjustbox}

%% file: numerical/quadrotor.tex
\begin{tikzpicture}
 
\node (input) at (-3.97,0) {};

\draw [dashed, fill=teal!20] (-2.8,-1) rectangle (2.8,1);
 
\node [draw,
    fill=blue!20,
    text width=1.5cm, align=center,
    minimum height=1.2cm,
    right=1.4cm of input
]  (pid) { \footnotesize PID
controllers};

\node [draw,
    fill=red!20, 
    text width=1.5cm, align=center,
    minimum height=1.2cm,
    right=1.4cm of pid
] (quad) {\footnotesize 6DoF quadrotor dynamics};

\node[right=1.4cm of quad
] (output) {};

\node (system) at (0, 1.4) {\footnotesize Unknown simulator system};

\draw[-stealth] (pid.east) -- (quad.west) 
    node[midway,text width=1cm, align=center, above]{\scriptsize motor RPMs};
\draw[-stealth] (input) -- (pid.west) 
    node[near start,text width=1.2cm, align=center, above]{\scriptsize velocity commands};
\draw[-stealth] (quad.east) -- (output) node[near end, text width=1.2cm, align=center, above]{\scriptsize position \& orientation};     

\end{tikzpicture} 

%% file: numerical/figs/plot1.tex
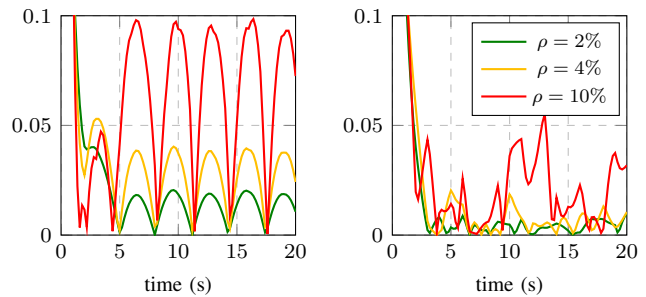
\begin{figure}[!ht]
\centering
    \begin{subfigure}[b]{0.49\columnwidth}
    \centering
    \begin{tikzpicture}[scale=1]
    \begin{axis}[
            legend style={font=\footnotesize},
            xlabel={time (s)},
            label style={font=\footnotesize},
            tick label style={font=\footnotesize},
            xmin=0, xmax=20,
            ymin=0, ymax=0.1,
            width=4.7cm, height = 4.5cm, 
            xtick={0,5,10,15,20},
            ytick={0,0.05,0.10},
            yticklabel style={/pgf/number format/fixed},
            legend image post style={scale=0.5},
            legend pos=north east,
            xmajorgrids=true,
            ymajorgrids=true,
            grid style=dashed,
    ]
    \addplot [mygreen, thick] table [x=iter, y=rho1, col sep=space] {mass_poi.csv};
    \addplot [myyellow, thick] table [x=iter, y=rho2, col sep=space] {mass_poi.csv};
    \addplot [red, thick] table [x=iter, y=rho3, col sep=space] {mass_poi.csv};
    
    \end{axis}
    \end{tikzpicture}
    \caption{Perturbations via Alg.~\ref{alg: attack}.}\label{fig: mass attack}
    \end{subfigure}
    \hfill
    \begin{subfigure}[b]{0.49\columnwidth}
    \centering
    \begin{tikzpicture}[scale=1]
    \begin{axis}[
            legend style={font=\footnotesize},
            xlabel={time (s)},
            label style={font=\footnotesize},
            tick label style={font=\footnotesize},
            xmin=0, xmax=20,
            ymin=0, ymax=0.1,
            width=4.7cm, height = 4.5cm, 
            xtick={0,5,10,15,20},
            ytick={0,0.05,0.10},
            yticklabel style={/pgf/number format/fixed},
           legend style={nodes={scale=0.9, transform shape}},
            legend pos=north east,
            xmajorgrids=true,
            ymajorgrids=true,
            grid style=dashed,
    ]
    \addplot [mygreen, thick] table [x=iter, y=rho1, col sep=space] {mass_rand.csv};
    \addlegendentry{$\rho=2\%$}
    
    \addplot [myyellow, thick] table [x=iter, y=rho2, col sep=space] {mass_rand.csv};
    \addlegendentry{$\rho=4\%$}
    
    \addplot [red, thick] table [x=iter, y=rho3, col sep=space] {mass_rand.csv};
    \addlegendentry{$\rho=10\%$}
    
    \end{axis}
    \end{tikzpicture}
    \caption{Random perturbations.}\label{fig: mass random}
    \end{subfigure}
    \caption{The tracking error of the first mass's displacement in the oscillating masses system.}\label{fig: masses}
\end{figure}

%% file: numerical/figs/plot2.tex
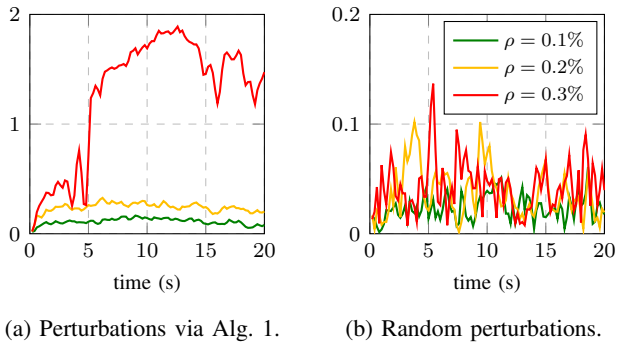
\begin{figure}[!ht]
\centering
    \begin{subfigure}[b]{0.49\columnwidth}
    \centering
    \begin{tikzpicture}[scale=1]
    \begin{axis}[
            legend style={font=\footnotesize},
            xlabel={time (s)},
            label style={font=\footnotesize},
            tick label style={font=\footnotesize},
            xmin=0, xmax=20,
            ymin=0, ymax=2,
            width=4.7cm, height = 4.5cm, 
            xtick={0,5,10,15,20},
            ytick={0,1,2},
            yticklabel style={/pgf/number format/fixed},
            legend image post style={scale=0.5},
            legend pos=north east,
            xmajorgrids=true,
            ymajorgrids=true,
            grid style=dashed,
    ]
    \addplot [mygreen, thick] table [x=iter, y=rho1, col sep=space] {quad_poi.csv};
    \addplot [myyellow, thick] table [x=iter, y=rho2, col sep=space] {quad_poi.csv};
    \addplot [red, thick] table [x=iter, y=rho3, col sep=space] {quad_poi.csv};

    \end{axis}
    \end{tikzpicture}
    \caption{Perturbations via Alg.~\ref{alg: attack}.}\label{fig: quad attack}
    \end{subfigure}
    \hfill
    \begin{subfigure}[b]{0.49\columnwidth}
    \centering
    \begin{tikzpicture}[scale=1]
    \begin{axis}[
            legend style={font=\footnotesize},
            xlabel={time (s)},
            label style={font=\footnotesize},
            tick label style={font=\footnotesize},
            xmin=0, xmax=20,
            ymin=0, ymax=0.2,
            width=4.7cm, height = 4.5cm, 
            xtick={0,5,10,15,20},
            ytick={0,0.1,0.2},
            yticklabel style={/pgf/number format/fixed},
           legend style={nodes={scale=0.9, transform shape}},
            legend pos=north east,
            xmajorgrids=true,
            ymajorgrids=true,
            grid style=dashed,
    ]
    \addplot [mygreen, thick] table [x=iter, y=rho1, col sep=space] {quad_rand.csv};
    \addlegendentry{$\rho=0.1\%$}
    
    \addplot [myyellow, thick] table [x=iter, y=rho2, col sep=space] {quad_rand.csv};
    \addlegendentry{$\rho=0.2\%$}
    
    \addplot [red, thick] table [x=iter, y=rho3, col sep=space] {quad_rand.csv};
    \addlegendentry{$\rho=0.3\%$}
    
    \end{axis}
    \end{tikzpicture}
    \caption{Random perturbations.}\label{fig: quad random}
    \end{subfigure}
    \caption{The tracking error of the quadrotor position in PyBullet simulator.}\label{fig: quad}
\end{figure}

%% file: conclusion/conclusion.tex
\section{Conclusion}
\label{sec: conclusion}

We study online poisoning attack problems in DPC. We develop an efficient numerical method to compute the data poisoning attacks based on implicit differentiation. Our future directions include attacks with partial knowledge of the parameters in DPC, as well as parameter design methods for DPC, where poisoning attacks act as a subroutine to evaluate the worst-case performance of any given parameter.
